\documentclass[pra,showpacs,twocolumn]{revtex4-1}
\usepackage{graphicx}
\usepackage{amsmath}\usepackage{amsfonts}
\newtheorem{thm}{Theorem}
\newtheorem{lem}[thm]{Lemma}

\newenvironment{proof}[1][]{\par\vskip12pt\noindent%
     {\it Proof\ \ifx#1==\else{ #1\/:}\newline\fi\ }\bgroup}
     {\egroup\quad\strut\hfill\hbox{\rule[-1pt]{8pt}{8pt}}\vskip12pt}

\def\eqref#1{(\ref{#1})}
\def\secref#1{Sect.~\ref{sec:#1}}

\def\idty{{\leavevmode\rm 1\mkern -5.4mu I}} %  unit operator

\def\Rl{{\mathbb R}}\def\Cx{{\mathbb C}}

\def\norm #1{\Vert #1\Vert}

\def\bra #1{\langle #1\vert}
\def\ket #1{\vert #1\rangle}
\def\braket #1#2{\langle #1 \vert #2\rangle}
\def\ketbra #1#2{\vert #1\rangle \langle #2\vert}
\def\kettbra#1{\ketbra{#1}{#1}}
\def\tr{\mathop{\rm tr}\nolimits}
\def\abs#1{\vert#1\vert}

\mathchardef\ree="023C \mathchardef\imm="023D  % former \Re and \Im
\renewcommand{\Re}{\ree e}

\def\flip{{\mathbb F}}
\def\BB{{\mathcal B}}

\def\HH{{\mathcal H}}
\def\KK{{\mathcal K}}

\def\L#1{{\mathcal L}^{#1}} % mostly \L2

\def\Fock{\Gamma}
\def\dFock{{\mathrm d}\Fock}
\def\tot#1{^{\otimes#1}}
\def\adj{^*}
\def\adj{^\dagger} %PhysRev demand

\def\hatsigma{\widehat\sigma}
\def\rtsigma{\sqrt{\hatsigma}}

\def\mbr#1#2{\left\lceil#1{:}#2\right\rfloor}
\def\mbr#1#2{\ifx#1\cdot\xi\else#1\fi{\left[#2\right]}}

\def\dirintE{\int^\oplus\!\!dE\,}

\def\factmom{\widehat m}
\def\factcf{\widehat C}

\newcommand {\fexp} [1] {\exp \left( #1 \right)}
\newcommand {\fabsq}[1] {\left| #1 \right|^2}

\begin{document}
\title{Full Counting Statistics of Stationary Particle Beams}

\author{J. Kiukas}
\author{A. Ruschhaupt}
\author{R.F. Werner}

\affiliation{Inst. Theoret. Physik, Leibniz Universit\"at Hannover,
 Appelstr. 2, 30167 Hannover, Germany}

\begin{abstract}
We present a general scheme for treating particle beams as many particle systems. This includes the full counting statistics and the requirements of Bose/Fermi symmetry. In the stationary limit, i.e., for longer and longer beams, the total particle number diverges, and a description in Fock space is no longer possible. We therefore extend the formalism to include stationary beams. These beams exhibit a well-defined ``local'' counting statistics, by which we mean the full counting statistics of all clicks falling into any given finite interval.
We treat in detail a model of a source, creating particles in a fixed state, which then evolve under the free time evolution, and we determine the resulting stationary beam in the far field. In comparison to the one-particle picture we obtain a correction due to Bose/Fermi statistics, which depends on the emission rate. We also consider plane waves as stationary many particle states, and determine the distribution of intervals between successive clicks in such a beam.
\end{abstract}

\maketitle
\section{Introduction}
In the standard framework of quantum mechanics each type of systems is assigned a Hilbert space $\HH$. Preparations are given by density operators on $\HH$, and measurements are given by operator valued measures on some outcome set. The basic predictions of the theory are probabilities, i.e., the asymptotic relative frequency, to be determined in a sufficiently long run experiments with the same preparation and measurement. However, many experiments do not follow this standard scheme: instead one often prepares a beam of particles. On the measurement side one then uses detectors, whose readout is a time series of clicks. Instead of probabilities the basic experimental results are then count rates. There are many reasons to study the relationship between these two models of quantum experiments. We were led to it in the course of a project on the tomography of single photon sources. Clearly, in this case it is relatively easy to get beam-type data, in which, for example, the antibunching dip in the $g^{(2)}$-correlation function indicates single particle emission \cite{Milburn,Yuan}. On the other hand, the emission time of individual particles is not controlled in beam mode, so absolute click times are meaningless. Hence it is impossible to get a full tomography of the single particle state created by such a source on demand; beam data and single-shot data have to be combined.

It is clear that beams can be described as many-particle quantum systems, and it is natural to use the standard Fock space for Bosons and Fermions. This will indeed be our starting point. However, the Fock space arena is too restrictive for describing truly stationary beams, which necessarily contain infinitely many particles. The full counting statistics of stationary beams will be obtained by taking the limit of longer and longer finite beams in a suitable way. The resulting description will be easier than that of finite beams, which must implicitly always contain the particulars of switching the beam on and off. Demonstrating this, and providing a theoretical framework for stationary beams is the main goal of this paper. On the measurement side this involves observables, whose outcome space is a (usually infinite) set of counting events. They are constructed directly in terms of one-particle observables, and do not rely on the identification of field intensities and count rates made in quantum optics via the Glauber model \cite{Glauber}.
On the side of the sources we focus on the type of sources, which make the connection to the one-particle picture most apparent, and which best realize  the idea of a train of independent particles. Technically, these are gauge invariant quasi-free states. We are, of course, aware that that this excludes many interesting phenomena, particularly for photons. But in our ongoing research on the subject we found the quasi-free beams to be an important starting point for more complex situations, e.g. by conditioning.

In order to make the paper more self-contained, we include brief statements of the relevant prerequisites, like the theories of counting processes \cite{Daley,Benard,BenardMacchi}, Fredholm determinants \cite{SimonTrace}, and quasi-free states \cite{CCR,CAR}. Processes rather similar to the ones we find have been discussed in the mathematical literature under the heading of ``determinantal processes''\cite{Shoshnikov}. However, to the best of our knowledge, these were always considered to be stationary in space \cite{Fichtner,Lytvynov} rather than in time. Stationarity in time requires some extra work, e.g., bringing in covariant arrival time observables on the one particle level, and the second quantization of generalized (POVM) observables. In the end, an appropriate ``local trace class'' condition characterizing the particle sources can be formulated quite simply also in the time case. In this we see the main contribution of our paper, although we also hope that for some readers it will serve as an invitation to a more systematic view of quantum counting processes.

Our paper is organized as follows. In \secref{2quant}, we introduce some notation, explain the concept of a point process, describe the construction of general counting observables, and give an account of arrival time measurements. In \secref{qfree}, we introduce the class of quasi-free states in Fock space, and compute the full counting statistics for such states. We also give the essential technical result which makes it possible to extend the results to the local statistics for stationary beams. \secref{source} is devoted to a concrete way of going to the stationary limit starting from an explicit dynamical description of particle creation. In \secref{beams} we give a general scheme for describing local counting statistics for quasi-free stationary beams; in particular, we get explicit formulas for correlation functions. Finally, in \secref{examples}, we apply the results to the simplest example of a stationary beam, the plane wave viewed as a many-particle state. Often enough such an interpretation is suggested in textbook treatments of scattering solutions of the stationary Schr\"odinger equation for one-particle potential scattering. Here we take it literally, and as a bonus get correlations and waiting time distributions in such a beam.

\section{Counting Observables}\label{sec:2quant}
Let us fix some notation. Throughout, $\HH$ will be the Hilbert space of a single particle. By $\Fock_s(\HH)$ we denote the Fock space over $\HH$, with $s=+1$ for Bosons and $s=-1$ for Fermions. When it is irrelevant, or clear from the context, the index $s$ will be omitted. For an operator $A$ on $\HH$, we will denote by $\Fock_s(A)$ the operator, which on $N$-particle wave functions acts like the $N$-fold tensor power $A\tot N$. Clearly, $\Fock_s(AB)=\Fock_s(A)\Fock_s(B)$.

For operator valued or scalar valued measures $\mu$ we abbreviate the integral over a scalar function $f$ as
\begin{equation}\label{mbr}
    \mbr\mu f=\int\!\mu(dx)\ f(x),
\end{equation}
i.e., we use round brackets for the set function and square brackets for the integral, i.e., the expectation value functional in the case of a probability measure.

For observables (POVMs=``positive operator valued measures'') the discrete case, in which points have finite measure, is often used. Then $\mbr Ff=\sum_xf(x)F_x$, where $F_x$ are positive operators with $\mbr F1=\sum_xF_x\leq\idty$. The projection valued special case is characterized by $F_xF_y=\delta_{xy}F_x$, or in a form also valid in the continuous case: $\mbr F{fg}=\mbr F{f}\mbr F{g}$.

In most textbooks, observables are simply identified with self-adjoint operators $A$, which presupposes that $X\subset\Rl$ and takes $F$ as the spectral measure of $A$, so $A=\int\!F(dx)x=\mbr Fx$. The same measure also defines the one-parameter unitary group $\exp(itA)=\mbr F{e^{itx}}$ generated by $A$. Generators are second quantized by $A\mapsto\dFock(A)$, where $\Fock(\exp(itA))=\exp(it\dFock(A))$ for all $t$.
However, for the purpose of this paper it is much more appropriate to start from POVMs on some outcome space $X$, which need not be the real line. The point is that for the natural second quantization of such POVMs \cite{Wer89a} the outcome space also changes:  where the observable $F$ at the one-particle level gives the probability of outcomes $x\in X$, its second quantization $\Fock F$ will correspond to a measurement of $F$ on every particle, and the result of this measurement is a distribution of points in $X$. The probabilities for such measurement outcomes constitute a so-called  ``point process''.

\subsection{Point processes}
A point process is a probability measure on the space of outcomes, where each outcome is a collection of (not necessarily distinct) points in a set $X$. We can think of each outcome as a possibly infinite numbered list $(x(1),x(2),...)$ with $x(i)\in X$, with the understanding that the ordering of the elements is irrelevant but, in contrast to the set $\{x(1),x(2),...\}$, we do count the number of occurrences of each $x\in X$.
A way to express this compactly is to take as the outcome of a point process the  so-called {\it empirical measure} $$\xi=\sum_i\delta_{x(i)},$$ where $\delta_x$ denotes the point measure (with $\delta$-function density) at $x\in X$. A measure $\xi$ of this form is also called a counting measure, and is characterized among measures by the property that the measure of each set is an integer, namely the number of points $x_i$ in that set.
Using the bracket notation introduced above, we then have $\mbr\xi f=\sum_if(x(i))$ for an empirical measure $\xi=\sum_i\delta_{x(i)}$. Since this bracket is linear in $f$, we can use it to characterize the probability distribution of a point process by its Fourier transform, i.e., by the expectation $\langle\cdot\rangle$ of the function $\xi\mapsto\exp(i\mbr\xi f)$. This is called the {\it characteristic function}
\begin{equation}\label{countcf}
    C(f)=\left\langle e^{i\mbr\xi f}\right\rangle
\end{equation}
of the distribution and contains the full counting statistics. For example, consider $k$ disjoint subsets $X_\ell$ of $X$, and let $\chi_\ell$ denote the indicator function of $X_\ell$. Then
we get the characteristic function of the joint probability distribution for the number of counts in the sets $X_\ell$ as
\begin{eqnarray}
    \sum_{n_1,\ldots n_k} p(n_1,\ldots n_k) e^{i\sum_\ell \lambda_\ell n_\ell} &=& \left\langle e^{i\sum_\ell \lambda_\ell \xi(X_\ell)}\right\rangle\nonumber\\
    &=& C({\textstyle\sum_\ell \lambda_\ell\chi_\ell}).\label{jointcounts}
\end{eqnarray}
If the particle numbers are independent for every partition of $X$ into sets $X_\ell$, we have a {\it Poisson process}, which is characterized by a measure $\mu$ on $X$, called the {\it intensity measure} of the process, such that $p(n_1,\ldots n_k)=\prod_\ell \mu(X_\ell)^n/(n!) \exp(-\mu(X_\ell)$ and hence,
\begin{equation}\label{poissoncf}
    C(f)=\exp\int\!\!\mu(dx)\bigl(e^{if}-1\bigr).
\end{equation}

The $k^{\rm th}$ {\it moment} of the point process is defined as the uniquely determined permutation symmetric measure $m_k$ on $X^k$, satisfying
$$
\int m_k(dx_1\,\cdots dx_k)\prod_{j=1}^kf(x_j) = \langle\ \xi[f]^k \rangle.
$$
Using arbitrary functions $f$ on $X^k$, we can equivalently give the definition as
\begin{equation}\label{momdef}
\mbr{m_k}f=\left\langle\ \sum_{i_1,\cdots, i_k}f(x(i_1),\ldots,x(i_k)) \right\rangle.
\end{equation}
Since
\begin{equation}\label{moment}
 C(f) = \sum_k\frac{i^k}{k!}\int m_k(dx_1\,\cdots dx_k)\prod_{j=1}^kf(x_j),
 \end{equation}
we can extract the moments from $C(\lambda f)$ by differentiating with respect to $\lambda$.
For a Poisson process, the expansion of the characteristic function in powers of $f$ is $C(f)=i\mbr\mu f-\frac12(\mbr\mu{f^2}+\mbr\mu{f}^2)+{\bf O}(f^3)$, so the first moment is $m_1=\mu$ and the second is $m_2=\mu\otimes\mu+\mu\circ\Delta^{-1}$ with the diagonal map $\Delta(x)=(x,x)$. The second moment thus has a singular part concentrated on the diagonal. This is not a special feature of the Poisson process, but occurs for any counting process. It is therefore customary to consider a modified set of moments, called factorial moments \cite{Daley, Milburn}, or ``correlation functions'' \cite{Shoshnikov}, which do not have such singularities. Like the moment $m_k$, the factorial moment of order $k$, which we denote by $\factmom_k$, is a permutation symmetric measure on $X^k$. For a function $f$ of $k$ variables, the factorial moment is defined by the following expectation:
\begin{equation}\label{factmomdef}
   \mbr{\factmom_k}f=\left\langle\ \sum_{\genfrac{}{}{0pt}{1}{i_1,\cdots i_k}{ {\rm distinct}}}
        f(x(i_1),\ldots,x(i_k))\right\rangle.
\end{equation}
By comparing the expression \eqref{factmomdef} to \eqref{momdef} it is clear that the exclusion of multiply occurring indices in the former just has the effect of eliminating the singular term from the second moment. For the Poisson process one has  $\factmom_k=\mu\tot k$ for all $k$.

The factorial moments are most easily obtained from the characteristic function $C$ by observing that the generating function
\begin{equation}\label{factcf}
 \factcf(f) = \sum_k\frac1{k!}\int\factmom_k(dx_1\,\cdots dx_k)\prod_{j=1}^kf(x_j)
 \end{equation}
is related to $C$ just by a transformation of the argument:
\begin{equation}\label{factcfcf}
  C(f) = \factcf(e^{if}-1)
\end{equation}
As an example of using \eqref{factcfcf}, consider the probability $p_Y(n)$ of finding exactly $n$ particles in a measurable region $Y\subset X$. By \eqref{jointcounts} and \eqref{factcfcf}, we get the relation
\begin{equation}\label{numberdist}
    \sum_{n=0}^\infty p_Y(n)z^n
       = \factcf((z-1)\chi_Y)
\end{equation}
for $z=e^{i\lambda}$, where $\chi_Y$ is the indicator function of $Y$.
In particular, we get the no event probability directly from the factorial moment generating function by analytic continuation:
\begin{equation}\label{numbergen}
    p_Y(0)=\factcf(-\chi_Y)\ .
\end{equation}

These probabilities determine the {\it interval statistics} of a point process on the time axis $X=\Rl$. Indeed, let $p_0(t_1,t_2)=p_{[t_1,t_2)}(0)$ denote the probability for not finding a click in the interval $[t_1,t_2)$. Then the probability of having no click on $[t_1,t_2)$ and at least one click just before $t_1$, say in the interval $[t_1-\varepsilon,t_1)$, is the same as having no click on $[t_1,t_2)$ and at least one in $[t_1-\epsilon,t_2)$, which is $p_0(t_1,t_2)-p_0(t_1-\varepsilon,t_2)$. Hence the conditional probability for having no click on $[t_1,t_2)$, on the condition of having at least one in $[t_1-\epsilon,t_1)$, is
$$ \frac{p_0(t_1,t_2)-p_0(t_1-\varepsilon,t_2)}{1-p_0(t_1-\varepsilon,t_1)}\ .$$
At the limit $\varepsilon\to0$, this tends to $1-$ the probability of having to wait at most time $\tau =t_2-t_1$ for the next click, after a click at $t_1$. Hence the probability density $w_{s}(\tau)$ for the waiting time $\tau\in [0,\infty)$ given a click at $s$ is
\begin{equation}\label{nextclickdensity}
    w_s(\tau)=-\left(\frac{\partial p_0}{\partial t_1}(s,s)\right)^{-1}\frac{\partial^2p_0}{\partial t_1\partial t_2}(s,s+\tau).
\end{equation}

Since we will eventually apply the above formalism to particle detection processes, we close this subsection by a remark on the role of the factorial moment densities in the standard theory of photon counting used in quantum optics. If $X$ has a natural measure $dx$ (typically $X=\Rl^m$ with the Lebesgue measure), we can often write $\factmom_k(dx_1\ldots dx_k)= h^{(k)}(x_1,\ldots,x_k)dx_1\ldots dx_k$ for a density function $h^{(k)}$. In Glauber's model of a photon detection process, we have $$h^{(n)}(x_1,\ldots,x_n) = {\rm const}\,\cdot G^{(n)}(x_1,\ldots,x_n,x_n,\ldots,x_1),$$ where $G^{(n)}$ is the usual ``correlation function'' defined using the field operators \cite{Glauberbook,Milburn}. In this context, one also typically uses the \emph{normalized correlation functions}, which we define for a general point process by
\begin{equation}\label{g2}
    g^{(n)}(x_1,\ldots,x_n)= \frac{h^{(n)}(x_1,\ldots,x_n)}{h^{(1)}(x_1)\cdots h^{(1)}(x_n)}.
\end{equation}

\subsection{Second quantization of general observables}
One reason for introducing characteristic functions of point processes is that they make the construction of the second quantized observable $\Fock F$ from the single particle observable $F$ extremely simple \cite{Wer89a}. Indeed, if we just express the idea that $\Fock F$ measures $F$ on all the particles we get, restricted to the $N$-particle space, the operator
\begin{align*}
\left(\mbr{(\Fock F)}{e^{i\mbr\cdot f}}\right)_N
   & =\int\!\! F(dx_1)\otimes\cdots\otimes F(dx_N)e^{\sum_if(x_i)}\\
    &=\left(\mbr F{e^{if}}\right)\tot N.
    \end{align*}
Taking the direct sum over $N$, we get the fundamental formula
\begin{equation}\label{charFockF}
    \mbr{(\Fock F)}{e^{i\mbr\cdot f}}
    =\Fock\Bigl(\mbr F{e^{if}}\Bigr).
\end{equation}

This expression makes sense on full Fock space, i.e., it does not require restriction to the Bose or Fermi sector. Hence, for a state given by a density operator
$\rho$ on this space the full counting statistics can be extracted from the characteristic function
\begin{equation}\label{charfunc}
    C(f)=\tr\rho \mbr{(\Fock F)}{e^{i\mbr\cdot f}}= \tr\rho \Fock(\mbr F{e^{if}}).
\end{equation}
Of course, for a Bose or Fermi system, the operator $\rho$ has support in the appropriate subspace and we can replace the $\Gamma$ in this formula by the corresponding restriction $\Fock_s$.

The factorial moments can be computed from this using Eqs. \eqref{factcfcf} and \eqref{factcf}. For a state $\rho$ on full Fock space we get
\begin{eqnarray}\label{factopcf}
    \factcf(f)&=&\tr\rho\Fock(\idty+\mbr Ff)\nonumber\\
       &=&\sum_{N=0}^\infty\tr\rho_N (\idty+\mbr Ff)\tot N \nonumber\\
       &=&\sum_{k=0}^\infty\tr\hat\rho_k \mbr Ff\tot k \ ,
\end{eqnarray}
where $\hat\rho_k$ is the reduced $k$-particle reduced density operator
\begin{equation}\label{rhohat}
    \hat\rho_k=\sum_{N=k}^\infty \genfrac{(}{)}{0pt}{0}{N}{k}\tr_{[k+1,\ldots,N]}\rho_N.
\end{equation}
Hence the factorial moments are
\begin{equation}\label{mhatgeneral}
   \factmom_k(dx_1\cdots dx_k)=k!\tr\hat\rho_k F(dx_1)\otimes\cdots F(dx_k).
\end{equation}

\subsection{Arrival time observables and their dilations}%
\label{sec:timeobs}
Due to an old argument of Pauli, an arrival time observable cannot be a spectral measure of a self-adjoint ``time operator''. However, the generalization of the notion of observables to POVMs immediately allows time-shift covariant observables to be constructed \cite{Kijowski}. In this subsection we describe the general construction of observables, which measure the arrival time $t$ and arrival location $x$ of a particle \cite{Wer86}. Here location is taken in a rather broad sense, and could just be the number of the detector which responds. We consider arbitrary observables, which are covariant for time translations, i.e.,
\begin{equation}\label{timecov}
    \exp(iHt)F[f]\exp(-iHt)=F[\tau_tf],
\end{equation}
where $H$ is the Hamiltonian, and $\tau$ is the time shift on functions of $t$ and $x$, i.e., $(\tau_tf)(x,t')=f(x,t'-t)$. The standard method \cite{Wer86} to build all covariant observables (even for a general covariance group with representation $g\mapsto U_g$) involves two steps: one first uses the Naimark dilation theorem to turn any generalized (POVM) observable into a projection valued one, say $\widetilde F$, which lives on another Hilbert space $\widetilde\HH$ and is connected to $F$ by an isometry $V:\HH\to\widetilde\HH$ so that $\mbr Ff=V^*\mbr{\widetilde F}fV$. There is also a unitary group representation $\widetilde U$ on $\widetilde\HH$, for which $\widetilde F$ is covariant, and which is intertwined by $V$, i.e., $\widetilde U_gV=VU_g$. In the second step one uses the theory of Mackey \cite{Mackey} who called projection valued covariant observables ``systems of imprimitivity'' and showed their intimate connection to induced representations. This second part is easy for just the time translation group $\Rl$, and leads to standard Schr\"odinger pairs of position and momentum operators, with some multiplicity. Thus in the dilation space ``energy'' is the canonical multiplication operator canonically conjugated to ``time'' and has therefore purely absolutely continuous spectrum. The covariant time observable approach is therefore limited to Hamiltonians $H$ with absolutely continuous spectrum, in which case the Hilbert space is of direct integral form
\begin{equation}\label{dirintE}
    \HH=\dirintE\HH_E.
\end{equation}
This is shorthand for the space of wave functions, which are functions of energy such that $\psi(E)\in\HH_E$, the multiplicity space at $E$. This will, of course, be $\{0\}$ when $E$ is not in the spectrum of $H$ (e.g. when $E<0$ for the standard kinetic energy). Scalar products are computed as
\begin{equation}\label{scpE}
    \braket\phi\psi_{\mathcal{H}}=\int dE\ \braket{\phi(E)}{\psi(E)}_{\mathcal{H}_E}
\end{equation}
with the scalar product of $\HH_E$. The technical (measurability) conditions on direct integral Hilbert spaces are to ensure that this expression makes sense (see e.g. \cite{Kadison}). Of course, the Hamiltonian is the multiplication operator $(H\psi)(E)=E\psi(E)$ in this representation. More generally, a bounded operator $A$ commutes with $H$ if $(A\psi)(E)=A_E\psi(E)$ for some measurable family of operators $A_E\in\BB(\HH_E)$. We write for this
\begin{equation}\label{dirintA}
    A=\dirintE A_E.
\end{equation}

If a space \eqref{dirintE} allows a projection valued covariant time observable, and hence a self-adjoint conjugate time operator, this operator generates a unitary group which shifts the energy variable. It thus introduces a canonical identification between all the spaces $\HH_E$. In particular, they must be non-zero also for negative energies, which is exactly the above-mentioned argument of Pauli that semi-bounded Hamiltonians do not allow a projection valued time observable. Nevertheless, this structure appears as the dilation of any given time observable. The Hilbert space in that case can be written either as the tensor product $\L2(\Rl,dE)\otimes\KK$ or, in the spirit of \eqref{dirintE}, as the space of $\KK$-valued $\L2$-functions on $\Rl$. The time observable in this case is computed in the usual way by Fourier-transforming to $\L2(\Rl,dt)\otimes\KK$, and the joint measurement of $t$ and $x$ is realized in this tensor product.
It is thus characterized by the following data:
\begin{enumerate}
\item a Hilbert space $\KK$, which will be the energy-independent multiplicity space of the dilated observable,
\item a family of isometries $V_E:\HH_E\to\KK$, which together define the dilation
isometry $V:\HH\to\L2(\Rl,dE)\otimes\KK$ via $(V\psi)(E)=V_E\psi(E)$, and
\item an observable $G$ with outcome space $X$ in the Hilbert space $\KK$.
\end{enumerate}
We have to compute the expectation operator $\mbr Fh$ for arbitrary functions $f$ of $(t,x)$ but it suffices to do this for the product functions $f(t,x)=h(t)g(x)$. For these the above data determine the operator
\begin{equation}\label{Fh}
    \mbr F{hg}\psi(E)=\int\!\! dE'\ \widehat h(E-E')\ V_E^*\ \mbr Gg V_{E'}\psi(E'),
\end{equation}
where $\widehat h$ denotes the Fourier transform of $h$, normalized as
\begin{equation}\label{fourierf}
    \widehat h(E)=\frac1{2\pi}\int dt\ e^{iEt} h(t).
\end{equation}
This ensures that for $g=1$ and $h\nearrow1$ we find $F(h)\nearrow\idty$, so the observable $F$ is normalized. It is convenient to allow also subnormalized observables, i.e., $F(1)\leq\idty$. In that case the operator $\idty-F(1)$ measures the probability that the particle never arrives. By construction this operator will commute with $H$, and the only modification in the above setup is to allow $V$ to be a general operator with $\norm V\leq1$, rather than an isometry.

\section{Quasi-free states}\label{sec:qfree}
\subsection{Physical background}
Let us begin with a Boltzmann statistical model of a multi-particle preparation: Suppose we have a one-particle preparation with density operator $\sigma_0$, which we run at $N$ random times $t_i$. Hence, if $\sigma(t)=e^{-iHt}\sigma_0 e^{iHt}$ is the time translate of $\sigma_0$, we get the state $\bigotimes_{i=1}^N\sigma(t_i)$. We ignore for the moment the symmetrization requirements, so we apply the observable $F$ to each of these systems separately, obtaining a point $x_i$ as a measuring result. In order to determine the characteristic function of the counting statistics we need the distribution of the emission times, which we take to be Poisson with intensity measure $\mu$ (i.e., with characteristic function $C_{\rm time}(g)=\exp\int\mu(dt)(e^{ig}-1)$, see \eqref{poissoncf}). From this we get the characteristic function of the counts $x_i$ as
\begin{equation}\label{cfBoltz}
    C(f)=\exp\int\mu(dt)\tr\sigma(t)F(e^{if}-1).
\end{equation}
This is again Poisson, and depends only on the integral $\sigma=\int\mu(dt)\sigma(t)$. The corresponding state on full Fock space is $\bigoplus_N\frac1{N!}\sigma\tot N$, where the factorial is the usual correction factor for indistinguishability familiar in classical statistical mechanics.

Of course, this state is not consistent with Bose or Fermi statistics. Its closest analogue is to replace $\frac1{N!}\sigma\tot N$, by $P_s\sigma\tot NP_s$, where $P_s$ denotes the projection onto the (anti-)symmetric subspace of $\HH\tot N$. That is, we consider the {\it quasi-free state} with density operator
\begin{equation}\label{qfree}
   \rho= \frac{\Fock_s(\sigma)}{\tr\Fock_s(\sigma)}
\end{equation}
on $\Fock_s\HH$. Here we do not take $\sigma$ to be normalized. Instead the normalization factor of $\sigma$ determines the particle number distribution. To be precise, \eqref{qfree} is a ``gauge invariant'' quasi-free state. More general quasi-free states, which do not necessarily commute with particle number are defined in \cite{BraRob2}, and will not be studied in this paper. Quasi-free states are plausible models for non-interacting particle beams, because they give the same results as Boltzmannian independence in the weak beam limit. They also describe naturally the production of a beam. Consider an oven, modeled  as an ideal gas with one-particle Hamiltonian $H$, at temperature $T=1/(k\beta)$ and chemical potential $\mu$. Then we have the grand canonical quasi-free state with $\sigma=\exp(-\beta(H-\mu\idty))$. We then get a beam by letting some particles escape through a hole, and we can also add (possibly time-dependent) one-particle potentials, collimating filters and the like. The important point is that as long as we only apply one-particle operations, i.e., unitary operators of the form $\Fock(U)$, the quasi-free character of the initial state will be preserved.

In the context of quantum optics this type of photon beam is usually called thermal light (see e.g. \cite{Milburn}), the state appearing as a special case of ``chaotic state'' \cite{Glauberbook}. The latter is defined in terms of the occupation number states $|\{ n_k\}\rangle$ as
\begin{equation}\label{chaotic}
\rho = \sum_{\{n_k\}} \prod_k \frac{\alpha_k^{n_k}}{(1+\alpha_k)^{1+n_k}} |\{ n_k\}\rangle\langle\{ n_k\}|,
\end{equation}
where $k$ indexes the modes, and the $\alpha_k>0$ are parameters satisfying $\sum_k \alpha_k<\infty$, each $\alpha_k$ coinciding with the expectation value of the mode $k$ occupation number. In fact, any quasi-free state of the form \eqref{qfree} for $s=1$ can be written as \eqref{chaotic} by choosing the modes according to an eigenbasis of the positive trace class operator $\sigma$; the parameters $\alpha_k$ are then the eigenvalues of $\hatsigma= \sigma / (\idty-\sigma)$.

\subsection{Characteristic functions}
A crucial tool in the following is a formula for the denominator in \eqref{qfree}. When $A$ is trace class (i.e., $\norm A_1=\tr\sqrt{A\adj A}<\infty$), and, in the Bose case $\norm A<1$, then $\Fock_s(A)$ is also trace class and
\begin{equation}\label{traceFock}
    \tr\Fock_s(A)=\det(\idty-sA)^{-s}.
\end{equation}
For the theory of such infinite dimensional determinants we refer to \cite{SimonTrace}. Using this formula, we get a simple expression for the characteristic function of a counting measurement:
\begin{eqnarray}\label{computecf}
C(f) &=& \frac{\tr{\Fock (\sigma)\Fock\left(\mbr F{e^{if}}\right)}}{\tr{\Fock (\sigma)}}=
\frac{\tr{\Fock\left(\sigma \mbr F{e^{if}}\right)}}{\tr{\Fock(\sigma)}}\nonumber\\
&=&\frac{\det\left(\idty-s\sigma \mbr F{e^{if}}\right)^{-s}}{\det(\idty-s\sigma)^{-s}}
\nonumber\\
&=&\det\left(\idty-s(\idty-s\sigma)^{-s}\sigma \left(\mbr F{e^{if}}-\idty\right)
\right)^{-s} \nonumber
\end{eqnarray}
To summarize:
\begin{eqnarray}\label{ar3}
C(f) &=&\det\left(\idty-s\hatsigma \mbr F{e^{if}-1}\right)^{-s},
\\\mbox{with}\qquad\label{hatsigmadef}
    \hatsigma&=&\frac\sigma{\idty-s\sigma},
\end{eqnarray}
where we took the liberty to write a fraction because numerator and denominator commute, and the expression can be evaluated in the functional calculus. It is useful to note the bounds on the operators $\sigma,\hatsigma$ in the Bose and Fermi case: Clearly both operators must be positive and have finite trace. In the Bose case we need in addition that $\sigma\leq(1-\varepsilon)\idty$ for some $\varepsilon>0$, which is equivalent to saying that $\hatsigma$ is bounded. In the Fermi case it is the other way around: $\sigma$ can be any bounded operator, which implies that $\hatsigma$ is strictly less than the identity. The formula \eqref{ar3} contains the complete counting statistics for the counting observable (compare also \cite{klich}).

\subsection{Factorial moments}

From \eqref{ar3} and \eqref{factcfcf} we get the factorial moment generating function
\begin{equation}\label{factcfqf}
    \widehat{C}(f) =  \det\left(\idty-s\hatsigma \mbr F{f}\right)^{-s}=\tr\Fock_s(\hatsigma \mbr F{f})
\end{equation}
Comparing this with \eqref{factopcf} we see that
the $k$-particle reduced density operators of the quasi-free state are
   $\hat\rho_k=P_s\hatsigma\tot kP_s$,
so by \eqref{mhatgeneral}, the factorial moments are given by
\begin{equation}\label{factmomqf}
 \factmom_k(dx_1\cdots dx_k)=k!\tr P_s\hatsigma\tot kP_s F(dx_1)\otimes\cdots F(dx_k).
 \end{equation}

The first moment is simply $\factmom_1(dx)=\tr \hatsigma F(dx)$; for the second moment, the expression can be further reduced, so that traces have only to be taken in the one particle space. To this end we write the (anti-)symmetrization projection $P_s=(\idty+s\flip)/2$, where $\flip$ is the unitary transposition operator, and use $\tr(\flip A\otimes B)=\tr(AB)$. Then
\begin{eqnarray}\label{factmomqf2}
   \factmom_2(dx\,dy)&=&\factmom_1(dx)\factmom_1(dy)
       \nonumber\\ &&\quad
       + s\, \tr( \hatsigma F(dx)\hatsigma F(dy)).
\end{eqnarray}
Now the trace on the right hand side is a positive measure on $X\times X$ and is also positive definite in the sense that it gives positive expectation to functions of the form $\overline{f(x)}f(y)$. This shows that for Bosons we always have the bunching effect $g^{(2)}\geq1$ and the antibunching effect $g^{(2)}\leq1$ for Fermions. (Recall the definition \eqref{g2} of the correlation function $g^{(2)}$). Clearly, there are interesting cases of photon antibunching, but these require artfully correlated, not quasi-free sources.

For later use we note the $k^{\rm th}$ order generalization of \eqref{factmomqf2}. The expression for $\tr( V_\pi A_1\otimes\cdot A_k)$, for a permutation operator $V_\pi$ is based on the cycle decomposition of the permutation $\pi$, say $\pi=(i_1,\ldots,i_r)(j_1,\ldots,j_s)\cdots$, and gives the product of the traces $\tr(A_{i_1}\cdots A_{i_r})\tr(A_{j_1}\cdots A_{j_s})\cdots$. It is convenient to introduce the measures
\begin{equation}\label{mu-ell}
    \mu_\ell(dx_1\cdots dx_\ell)=\tr\left(\prod_{\alpha=1}^k \hatsigma F(dx_\alpha)\right),
\end{equation}
so that $\mu_1(dx)=\factmom_1(dx)$, and \eqref{factmomqf2} reads $\factmom_1(dx\,dy)=\mu_1(dx)\mu_2(dx)+s\mu_2(dx\,dy)$. Then, for example, for $k=3$, we find
\begin{eqnarray}\label{factmomqf3}
   \factmom_3(dx\,dy\,dz)&=&\mu_1(dx)\mu_1(dy)\mu_1(dz)
       \nonumber\\ &&\quad
       + s\, \bigl(\mu_2(dx\, dy)\mu_1(dz) +\mbox{cyclic}\bigr)
       \nonumber\\ &&\quad
       + 2\Re\mu_3(dx\,dy\,dz).
\end{eqnarray}
For general $k$ we get similar expansions into products of measures, the combinatorics of which requires some representation theory of the permutation group, which we will not expound here.

\subsection{The localization Lemma}\label{sec:localization}
The main aim of our paper is to establish {\it local} counting statistics even in situations where the global particle count is infinite, as will be the case for any stationary beam, or translationally invariant gas. The idea is to use the formula \eqref{ar3} for the characteristic function even in situations, where the operator $\hatsigma$ has infinite trace, but the product $\hatsigma F$ is sufficiently well behaved so the formula makes sense as written. Actually, even more general situations can be covered, if we replace $\hatsigma F$ by $\rtsigma F\rtsigma$. That is we start from the characteristic function
\begin{equation}\label{cfroot}
    C(f) =\det\left(\idty-s\rtsigma\mbr F{e^{if}-1}\rtsigma\right)^{-s}.
\end{equation}
This is the same as \eqref{ar3} when $\hatsigma$ has finite trace. Indeed, we have the identity \cite{ReeSim4det} $\det(\idty+AB)=\det(\idty+BA)$, whenever both $AB$ and $BA$ are both trace class. In the case at hand this is applied to the two Hilbert-Schmidt operators $A=\rtsigma$ and $B=\mbr F{e^{if}-1}\rtsigma$. The form given in the following Lemma is even slightly more general.

\begin{lem}\label{loclemma} Let $F$ be a measure on some set $X$, whose values are positive operators on a Hilbert space $\HH_1$, with $F(X)\leq\idty$, and consider a subset $X_0\subset X$. Let $\HH_2$ be another Hilbert space and $W:\HH_2\to\HH_1$ a bounded operator such that
\begin{equation}\label{WsWbound}
    \tr W^*F(X_0)W<\infty,
\end{equation}
and in the Fermi case ($s=-1$) also $\norm W<1$. Then the formula
\begin{equation}\label{cfqfW}
    C(f)=\det\left(\idty-sW^* \mbr F{e^{if}-1}W\right)^{-s},
\end{equation}
for all $f$ vanishing outside $X_0$ defines the characteristic function of a point process in $X_0$.
\end{lem}

\begin{proof}
Consider the Naimark dilation $F=V^*\widetilde FV$. Then for $f$ with support in $X_0$ we can write
\begin{eqnarray}\label{ww}
    W^* \mbr FfW&=& W^*V^*\mbr{\widetilde F}f VW \nonumber\\
        &=& W^*V^*\widetilde F(X_0)\mbr{\widetilde F}f \widetilde F(X_0)VW \nonumber\\
        &=&\widetilde W^*\mbr{\widetilde F}f \widetilde W, \nonumber
\end{eqnarray}
where $\widetilde W=\widetilde F(X_0)VW$, and at the second equality we used the projection valuedness of $\widetilde F$. Now by assumption \eqref{WsWbound} the operator $\widetilde W^*\widetilde W$ has finite trace, i.e., $\widetilde W$ is a Hilbert-Schmidt operator.
Moreover (relevant only for $s=-1$) $\norm{W}<1$ implies $\norm{\widetilde W}<1$, because $\widetilde F(X_0)$ is a projection, and the dilation operator
$V$ satisfies $V^*V=F(X)\leq\idty$. Hence $\widetilde W\widetilde W^*$ satisfies all conditions required of an operator $\hatsigma$ to define the counting statistics of a bona fide quasi-free state, with respect to a second quantized observable. The associated characteristic function \eqref{ar3} is
\begin{equation}\label{cfqfWtilde}
    C(f)=\det\left(\idty-s\widetilde W\widetilde W^* \mbr{\widetilde F}{e^{if}-1}\right)^{-s},
\end{equation}
which has the form stated in the lemma by the same argument that gave the equality of \eqref{ar3} and \eqref{cfroot} at the beginning if this section.
\end{proof}

Since we can take $W=\rtsigma$ in the Lemma, we find $\tr\rtsigma F(X_0)\rtsigma<\infty$ as a sufficient condition to apply \eqref{cfroot}. The similar looking condition $\norm{\hatsigma F(X_0)}_1<\infty$, which is suggested by the characteristic function \eqref{ar3}, is actually stronger. This is implied by the estimate $\tr \sqrt AB\sqrt A\leq\norm{AB}_1$, which holds for arbitrary positive operators $A,B$. (For a proof note that the trace norm is the sum of the singular values, which dominates the sum of the absolute values of the eigenvalues \cite[Thm. 1.15]{SimonTrace}, and that $AB$ and $\sqrt AB\sqrt A$ have the same nonzero eigenvalues.)

As a byproduct we can now approximate the counting statistics of a stationary state by the counting statistics of finite-beam ones:
\begin{lem}\label{loclimitlemma}Suppose that $W_0:\mathcal{H}_2\to\mathcal{H}_1$ satisfies the conditions of Lemma \ref{loclemma} for a given measure $F$ and a set $X_0\subset X$. Let $W_\infty,W_\alpha:\mathcal{H}_2\to\mathcal{H}_1$ be bounded operators, such that $W_\alpha^*\to W_\infty^*$ strongly, and
\begin{equation}\label{Wbound}
W_\alpha W_\alpha^*\leq W_0 W_0^*.
\end{equation}
Then $W_\infty$ and each $W_\alpha$ satisfy the conditions of Lemma \ref{loclemma}, and for the associated characteristic functions $C_\alpha, C_\infty$ holds
$$
C_\infty(f) = \lim_{\alpha} C_\alpha(f)
$$
uniformly for $f$ with support in $X_0$.
\end{lem}
\begin{proof}
Defining $\widetilde W_\alpha$ for $W_\alpha$ as in Lemma \ref{loclemma}, the assumption \eqref{Wbound} gives $\widetilde W_\alpha \widetilde W_\alpha^*\leq \widetilde W_0 \widetilde W_0^*$, which implies that the conditions of Lemma \ref{loclemma} are valid also for each $W_\alpha$ and $W_\infty$, and that $\|\widetilde W_\alpha^*\psi\|^2\leq \|\widetilde W_0^*\psi\|^2$ for $\psi\in \mathcal{H}$ where now $\widetilde W_0^*$ is Hilbert-Schmidt. Together with the strong convergence $\widetilde W_\alpha^*\to \widetilde W_\infty^*$, this implies convergence in the Hilbert-Schmidt norm. Hence, $\widetilde W_\alpha \widetilde W_\alpha^*\to \widetilde W_\infty\widetilde W_\infty^*$ in the trace norm. Using \eqref{cfqfWtilde}, we get the uniform convergence of the $C_\alpha$ from the estimate $|\det(\idty+A)-\det(\idty+B)|\leq \|A-B\|_{1}e^{\|A\|_{1}+\|B\|_{1}+1}$ \cite{ReeSim4det}.
\end{proof}

Now, in order to approximate a stationary beam with non-trace class $\hatsigma$, by a sequence of finite beams with trace class $\hatsigma_\alpha$, it is sufficient to have the weak operator convergence $\rtsigma_\alpha\to\rtsigma$, and majorization $\hatsigma_\alpha\leq \hatsigma_0$ by some $\hatsigma_0$ with $\tr\sqrt{\hatsigma_0} F(X_0)\sqrt{\hatsigma_0}<\infty$. This will give uniform convergence of characteristic functions for the local counting statistics.

\subsection{``Small $s$'' and parastatistics}
We end this section with a remark on parastatistics and weak beams. If we write \eqref{ar3} as
\begin{equation}\label{ar4}
C_s(f) =\exp\tr\left(\frac{-1}s\log(\idty-s\hatsigma \mbr F{e^{if}-1})\right),
\end{equation}
the formula gives corrected Boltzmann statistics \eqref{cfBoltz} for $s=0$, and parafermi (resp.~parabose) statistics of order $p$ for $s=-1/p$ (resp. $s=1/p$).
From \eqref{ar4} we get a rather uniform notion of weak beams: Whenever $\sigma$ is small (or in the parastatistics case: when $s$ is small), we nearly get Poisson statistics. In quantitative terms, the operator norm $\norm\hatsigma_\infty$ (the largest eigenvalue) measures well the maximal effect of statistics, in some sense the maximal phase space density:
\begin{eqnarray}\label{weakbeam}
    \abs{\log C_s(f)-\log C_0(f)}&\leq& \norm\hatsigma_1\,\beta(\abs s\norm\hatsigma_\infty)
        \nonumber\\ \mbox{with }\
    \beta(h)&=&-1-\frac{\log(1-h)}h\approx \frac h2.
\end{eqnarray}

\subsection{For comparison: coherent states}

Quasi-free states (which are associated to e.g. thermal light) are very different from coherent states, which are commonly used in describing laser beams. In order to make the distinction clear, we derive here the characteristic function and factorial moments for the latter case. The (Bose) \emph{coherent states} are given by the non-normalized vectors
\begin{equation}\label{e2thephi}
   e^\phi = \bigoplus_{N=0}^\infty \frac{1}{\sqrt{N!}}\phi\tot N,
\end{equation}
in the Bose Fock space $\Fock_+(\mathcal H)$. The characteristic function \eqref{charfunc} is now
\begin{eqnarray}\label{cfcoherent}
   C(f) &=& \left\langle e^\phi \mid \Fock_+(\mbr F{e^{if}})e^\phi\right\rangle / \| e^\phi\|^2
              \nonumber\\
        &=& \exp\braket \phi {F(e^{if}-1)\phi}.
\end{eqnarray}
This is precisely the characteristic function \eqref{poissoncf} of a Poisson random field with intensity measure
\begin{equation}\label{coherentPoisson}
    \mu(dx)=\braket\phi {F(dx)\phi}.
\end{equation}
It is remarkable that this holds for {\it any} second quantized observable.
The factorial moments are now simply products of $\mu$:
\begin{equation}\label{factmomcoh}
\factmom_k(dx_1\cdots dx_k) = \mu(dx_1)\cdots\mu(dx_k).
\end{equation}
Comparing with \eqref{factmomqf2}, we see that a quasi-free state with the same first moment, i.e., with
$\hatsigma= |\phi\rangle \langle \phi|$, has larger second moment. In particular, thermal light has more variance than coherent light of same intensity.

From \eqref{factmomcoh} we also immediately see that for any coherent state and any counting observable, all the correlation functions \eqref{g2} are constant,
$g^{(n)}(x_1,\ldots, x_k) = 1$, as expected \cite{Glauberbook,Milburn}.

\section{Stationary limit of a particle source}\label{sec:source}\noindent
The kind of limit we have to take here is clear from the case of the Boltzmann statistical model \eqref{cfBoltz}. We assumed there that $\mu$ is a finite measure, so the total number of particles had finite expectation. But we can also take $\mu(dt)=\mu_0dt$ as a multiple of Lebesgue measure, where $\mu_0$ is the emission rate, i.e.,
\begin{equation}\label{sigmaintegral}
    \sigma=\mu_0\int\!\! dt\ \sigma(t),
\end{equation}
provided the integral is a bounded operator and, for finite intervals $S$, $\tr\sqrt{\sigma} F(S)\sqrt{\sigma}$ has finite trace. One can take this as a motivation for looking integrated trace class operators \eqref{sigmaintegral} also in the case of Bose/Fermi statistics.
However,  it is physically more realistic to have a description of beam generation which is consistent with statistics from the outset.

For building a stationary source model it is best to include the particle generation in the dynamics. In this way one can consider sources operating continuously for an arbitrarily long time. The intuitive idea is that after being activated at time $t=0$, the source creates particles with fixed initial wave function $\phi$, one after another, each particle subsequently evolving according to some single particle Hamiltonian $H$ with direct integral decomposition as discussed in \secref{timeobs}. Formally, this is expressed (see, e.g., the review \cite{Alicki}, and \cite{Spohn}) by evolving the many particle state $\rho_t$ according to the master equation
\begin{equation}\label{mastereq}
    \frac{d}{dt}\rho_t=-i[d\Gamma(H),\rho_t]+\lambda(2a_\phi\rho_t a_\phi\adj -a_\phi a_\phi\adj \rho_t-\rho_t a_\phi a_\phi\adj ),
\end{equation}
with the initial condition that $\rho_0$ is the vacuum state. Here and in the rest of this section we consider only the Bosonic case. Then $\lambda>0$ quantifies the strength of the source, and $d\Gamma(H)$ is the many particle Hamiltonian corresponding to $H$. The time evolution $t\mapsto \rho_t$ is a quasi-free semigroup \cite{Alicki}, so the Fock space state $\rho_t$ is quasi-free for each $t\geq 0$. This reduces the dynamics to a one-particle problem, which has the solution
\begin{equation}\label{solvesigmat}
    \hatsigma_\phi(t) = 2\lambda \int_0^t\!\!ds\  e^{sT_\phi}\kettbra\phi \left(e^{sT_\phi}\right)^*
\end{equation}
where $T_\phi= -iH+ \lambda\kettbra\phi$ is the generator of a strongly continuous semigroup.

The integral $\hatsigma_\phi(\infty)$ is the analogue of \eqref{sigmaintegral}. To compute it, we can formally solve the function $\beta_\psi(t) = \Theta(t)\langle e^{tT_\phi}\phi |\psi\rangle$, in terms of $\gamma_\psi(t)=\Theta(t)\langle e^{-itH} \phi |\psi\rangle$, where $\Theta$ is the Heaviside step function. In fact, we get $\check\beta_\psi = S_\phi \check{\gamma}_\psi$, where $\check{h}$ denotes the inverse of the Fourier transform given in \eqref{fourierf}, and
\begin{equation}\label{statcorr}
S_\phi(E) =(1-\lambda \check\gamma_{\phi}(E))^{-1}.
\end{equation}
This gives $\langle \psi | \hatsigma_\phi(\infty)\psi\rangle = \lambda\pi^{-1}\int_{-\infty}^\infty |S_\phi(E)\check\gamma_{\psi}(E)|^2\, dE$, corresponding to a quasi-free state which is invariant with respect to the evolution \eqref{mastereq} but not with respect to the free Hamiltonian $H$.

In order to get a state which is stationary for the free evolution, we now take another limit $\hatsigma=\lim_{s\to+\infty}e^{iHs}\hatsigma_\phi(\infty)e^{-iHs}$. In the spirit of scattering theory, this amounts to translating the state back in time with the free evolution, after having let it evolve a long time according to the particle generating semigroup evolution.
Since $\gamma_{e^{-iHs} \psi}(t)= \gamma_\psi(t-s)$, we have $e^{isE}\check{\gamma}_{e^{-iHs}\psi}(E)\to2\pi \langle \phi(E)|\psi(E)\rangle$ as $s\to +\infty$, where $\phi(E)$ is the wave function $\phi$ in the $H$ (energy) representation \eqref{dirintE}. Here we have used the fact that
\begin{equation}\label{eitHrep}
\langle e^{-itH}\phi |\psi\rangle = \int e^{itE} \langle \phi(E)|\psi(E)\rangle_{\HH_E}.
\end{equation}
Thus the final stationary limit is given by
\begin{equation}\label{hatsiglimit}
    \hatsigma=4\pi \lambda\dirintE |S_\phi(E)|^2\kettbra{\phi_E}.
\end{equation}

In the 1D case with the free Hamiltonian, one can alternatively take a limit of moving the source to $-\infty$ in space, which results in a similar expression but with the projection onto positive momenta. Even in this case the denominator, which reflects the phase space density at the source, still depends on the negative momentum components of $\phi$.

There are, of course, some assumptions needed to make the above derivation valid. Physically, we expect \cite{Spohn} that the free evolution $H$ should be fast enough compared to the strength of the source so that the particles do not accumulate, or even condense, near the source, but move away as new ones are created. Mathematically,
the relevant assumptions can be expressed as follows: (i) The wave function $\phi$ is bounded in the energy representation, and (ii) the operator of multiplication by $S_\phi(E)$ defines a bounded operator which keeps the Hardy class $H^{2-}$ invariant. The assumption (i) ensures that $\|\gamma_{\psi}\|_2/\|\psi\|$ is uniformly bounded due to \eqref{eitHrep}; in particular, the $L^2$ Fourier transform $\check\gamma_{\phi}$ belongs to the Hardy class $H^{2-}$, and so can be extended to an analytic function in the open lower half plane. Then (ii) implies that $S_\phi \check\gamma_\psi$ is square integrable, and has support on $[0,\infty)$; hence $S_\phi$ is well defined, the formal relation $\check\beta_\psi = S_\phi \check\gamma_\psi$ makes sense, and \eqref{hatsiglimit} is bounded, the limits existing in the weak operator topology.

The assumption (ii) can be replaced with stronger but more easily verifiable versions: for instance, if $\gamma_{\phi}$ is (absolutely) integrable, and $\check\gamma_{\phi}(E) \neq \lambda^{-1}$ for all $E$ in the closed lower half plane, then (ii) holds. Even stronger condition \cite{Spohn} is $\lambda \int_0^\infty |\gamma_{\phi}(t)| \, dt<1$.

The assumption (ii) implies, in particular, that $\int |S_\phi(E)|^2\|\phi_E\|^2\, dE<\infty$. In the next section we will see in a more general context that this condition ensures $\tr\sqrt{\hatsigma}F[f] \sqrt{\hatsigma}<\infty$ for any arrival time observable $F$ as in \secref{timeobs}, and any $f$ compactly supported in the time direction. By Lemma \ref{loclemma}, the local counting statistics of $\Gamma_+F$ is therefore well-defined for the stationary state \eqref{hatsiglimit}, and can be obtained from \eqref{cfroot}.

In order to complete the discussion on this stationary limit, we have to check that the counting statistics for the limit state \eqref{hatsiglimit} can be approximated by measuring the actual finite beam emitted by the source. It is clear by construction that we can write $\hatsigma_{s,t}=e^{iHs}\hatsigma_\phi(t)e^{-iHs}=W_{s,t}^*W_{s,t}$, where $(W_{s,t}\psi)(E) = \sqrt{\lambda/\pi} e^{isE}\int_0^t e^{-iEt'} \beta_{\phi,e^{-isH}\psi}(t')\, dt'$. Here $W_{s,t}$ converge strongly when we take first $t\rightarrow \infty$ and then $s\rightarrow\infty$. Moreover, $\hatsigma_{s,t}\leq \|S_\phi\|^2\hatsigma_0$, where $\hatsigma_0$ is \eqref{hatsiglimit} without $|S_\phi(E)|^2$. Hence, it follows from Lemma \ref{loclimitlemma} that
\begin{equation}\label{spohncharlimit}
\lim_{s\rightarrow\infty}\lim_{t\rightarrow\infty} C_{s,t}(f) = \det(\idty - \rtsigma F[e^{if}-1]\rtsigma)^{-1},
\end{equation}
where $C_{s,t}$ is the characteristic function \eqref{cfroot} corresponding to the operator $\hatsigma_{s,t}$.

\section{Beams and rates}\label{sec:beams}
\subsection{Beams}
The following general scheme emerges from the above. We consider systems with Hamiltonian $H$, and decompose the Hilbert space into a direct integral $\HH=\dirintE\HH_E$ over the spectrum of $H$. Stationary beams are described by an extension of quasi-free states, given in terms of a one-particle operator
\begin{equation}\label{sigmageneral}
    \hatsigma=\dirintE\ \hatsigma(E),
   \end{equation}
commuting with the Hamiltonian. Here $\hatsigma(E)$ is a positive trace class operator in the multiplicity space $\HH_E$ at $E$, which depends on the details of the source. The basic normalization condition for these operators is that
\begin{equation}\label{sigmarate}
    \gamma=\frac1{2\pi}\int\!\!dE\ \tr\hatsigma(E)<\infty.
\end{equation}

We emphasize that no such operator is trace class. Indeed, a general trace class operator $T$ would have an integral kernel $T(E,E'):\HH_{E'}\to\HH_{E}$ so that
$(T\psi)(E)=\int\!\!dE'\ T(E,E')\psi(E')$. The trace of such an operator is $\tr T=\int\!\!dE\ \tr T(E,E)$. In contrast, the operator \eqref{sigmageneral} has the formal integral kernel $T(E,E')=\hatsigma(E)\delta(E-E')$, which is singular on the diagonal. Such operators do arise from integration of trace class operators over
time in the sense of \eqref{sigmaintegral}: The time evolved operator $U_t^*TU_t$ has integral kernel $T(E,E')\exp(it(E-E'))$. Integrating this with respect to time we get the kernel $2\pi\delta(E-E')T(E,E)$. Hence the trace class condition for $T$ turns into \eqref{sigmarate} for the integral $\int\!dt\, U_t^*TU_t$.

The direct integral form \eqref{sigmageneral} or, in other words, the elimination of off-energy-diagonal terms in the kernel for $\hatsigma$ is, of course, just the consequence of stationarity $[\hatsigma,H]=0$, and leads to a major simplification in the computation of expectation values and rates.

\subsection{Rates}
We now want to combine the stationary sources given by $\hatsigma$ of the form \eqref{sigmageneral}, with a general counting observable $\Fock_sF$, i.e., the second quantization of a general arrival time observable $F$, as discussed in \secref{timeobs}. The full \emph{local} counting statistics is then contained in the characteristic function $C(f)$ restricted to test functions $f$ which have compact support in the time direction. Technically, this approach is based on the discussion in \secref{localization}, which guarantees the existence of $C(f)$ in \eqref{cfroot} for any $f$ with $f(x,t) = 0$ for $t\notin [t_1,t_2]$,
once we have
\begin{equation}\label{ftrace}
{\rm tr} \sqrt{\hatsigma} F[\chi] \sqrt{\hatsigma}<\infty,
\end{equation}
where $\chi(x,t) = 1$ for $t\in [t_1,t_2]$ and $\chi(x,t) = 0$ otherwise.
One of the consequences of this approach is that we can always get a justification of the formula starting from finitely extended beams (trace class $\hatsigma$) and going to a stationary limit in \eqref{cfroot}. There are many ways to do such a limit, which corresponds to the many ways a beam which looks basically stationary during a fixed time interval could begin in the distant past and end in the far future. The formula \eqref{cfroot} thus captures the essence of what we mean by ``stationary beams''.

The rest of this subsection will be devoted to substantiating the above claim \eqref{ftrace} 
for $\hatsigma$ of the form \eqref{sigmageneral}, any time-covariant $F$, and any $t_1<t_2$.
We do this by showing the more general estimate
\begin{equation}\label{trnormrt}
    \norm{\rtsigma\mbr F{f}\rtsigma}_1\leq \|f\|_\infty \gamma \abs{t_2-t_1},
\end{equation}
for any bounded complex (measurable) test function $f$ such that $f(t,x)=0$ for $t\notin[t_1,t_2]$,
where $\gamma$ is the rate constant from \eqref{sigmarate}. First we apply the triangle inequality to a sum of positive operators, and use that on positive elements the trace norm is just the trace. That is for positive operators $F_\alpha$ and $f_\alpha\in\Cx$ we have
\begin{equation}
   \Bigl\Vert{\sum_\alpha f_\alpha F_\alpha}\Bigr\Vert_1
   \leq \max_\alpha\abs{f_\alpha}\sum_\alpha\tr F_\alpha
\end{equation}
Hence, for step functions $f= \sum_\alpha f_\alpha \chi_\alpha$, with $\chi = \sum_\alpha \chi_\alpha$, the left hand side of \eqref{trnormrt} is bounded by $\|f\|_\infty \tr\rtsigma\mbr F{\chi}\rtsigma$. To prove \eqref{trnormrt}, it is therefore sufficient to show that
\begin{equation}\label{trFest}
\tr\rtsigma\mbr F{\chi}\rtsigma\leq \gamma \abs{t_2-t_1}.
\end{equation}
We do this by expressing $F$ by its dilation $F=V^*\widetilde FV$. Since $\mbr{\widetilde F}{\chi}$ is a projection, the trace we need to compute is $\tr W^*W$ with
\begin{equation}\label{Ww}
    W=\mbr{\widetilde F}{\chi}V\rtsigma.
\end{equation}
Since $\tr W^*W=\tr WW^*$ for any operator (where both sides might still be infinite), we now compute $\tr WW^*=\mbr{\widetilde F}{\chi}V\hatsigma V^*\mbr{\widetilde F}{\chi}$. Note that
\begin{equation}\label{VsigVstar}
    V\hatsigma V^*=\dirintE V_E\hatsigma(E) V_E^*
\end{equation}
commutes with the energy. Therefore, in the time domain it acts as a convolution operator:
\begin{eqnarray}\label{St}
    (V\hatsigma V^*\Phi)(t)&=& \int\!ds\ S(t-s)\psi(s)
       \nonumber\\\mbox{with}\qquad
    S(t)&=&\frac1{2\pi}\int\! dE\ e^{itE}V_E\hatsigma(E) V_E^*.
\end{eqnarray}
Here the integral defining $S$ is convergent in trace norm by assumption \eqref{sigmarate}, and by the Riemann-Lebesgue Lemma $t\mapsto S(t)\in\BB(\KK)$ is a continuous function vanishing at infinity. Due to the continuity we can evaluate the trace as an integral on the diagonal of the kernel $K(t,s)=S(t-s)$ \cite[Thm. 3.9.]{SimonTrace}:
\begin{equation}
    \tr \mbr{\widetilde F}{\chi}V\hatsigma V^*\mbr{\widetilde F}{\chi}
      =\int_{t_1}^{t_2}\!\!\!dt\ \tr_\KK S(0).
\end{equation}
But
\begin{equation}\label{trS0}
    \tr_\KK S(0)=\frac1{2\pi}\int\!dE\ \tr\hatsigma(E) V_E^*V_E
               \leq\gamma.
\end{equation}
This completes the proof of the estimate \eqref{trnormrt}.

According to \secref{localization}, in order to approximate the counting statistics of a stationary state (described by $\hatsigma$ of the form \eqref{sigmageneral}) by the counting statistics of finite-beam ones (with trace class $\hatsigma_\alpha$), it is sufficient to find one majorizing $\hatsigma_0$ of the form \eqref{sigmageneral}, with $\hatsigma_\alpha\leq \hatsigma_0$, and have $\sqrt{\hatsigma_\alpha}\to \sqrt{\hatsigma}$ at least weakly.
This will give uniform convergence of characteristic functions for the local counting statistics.

\subsection{Moments}
The moments are best expressed in terms of the operator valued function $S(t)$, defined in \eqref{St}. This depends both on the source via $\hatsigma$ and on the observable chosen, via $V$. It also contains the required Fourier transformations, so all moments are immediately expressed in the time domain. The idea is to reduce the factorial moments to the measures from \eqref{mu-ell}, rewritten by using the dilation:
\begin{equation}
    \mu_\ell(dx_1\cdots dx_\ell)=\tr\left( V\hatsigma V^*\widetilde F(dx_1)\cdots V\hatsigma V^*\widetilde F(dx_\ell)\right).
    \nonumber
\end{equation}
The operator under the trace has integral kernel
\begin{eqnarray}
    K(t_0,t_\ell)=\int\!\! dt_1&\cdots&dt_{\ell-1} S(t_0-t_1)\widetilde F(dx_1)\times
                    \nonumber\\
        \times S(t_1-t_2)&\cdots& S(t_{\ell-1}-t_\ell)\widetilde F(dx_\ell)
                    \nonumber
\end{eqnarray}
Then the required trace is $\int\!dt\ \tr K(t,t)$, where the trace in the integrand is over $\KK$.
In all these expressions the argument of the measure $\widetilde F$ is still the combination of time and arrival location. But noting that with respect to time
$\widetilde F$ is just a multiplication operator, so for $f(t,x)=h(t)g(x)$ we have $(\mbr{\widetilde F}{hg}\Phi)(t)=h(t)\mbr Gg\Phi(t)$. Here and in the following $x$ stands only for the arrival location. If we now take $h$ and $g$ as the indicator functions of small sets $dt\subset\Rl$ and $dx\subset X$, we find that
\begin{eqnarray}\label{muellSt}
   &&\mu_\ell(dt_1\,dx_1,\cdots,dt_\ell\,dx_\ell)= dt_1\cdots dt_\ell \times\\
   &&\quad\times\tr\Bigl( S(t_\ell-t_1)G(dx_1)S(t_1-t_2)\cdots G(dx_\ell)\Bigr)
   \nonumber
\end{eqnarray}
For the first (factorial) moment we thus get
\begin{equation}\label{timemoment1}
   {\factmom_1(dt\,dx)}= dt\,\tr S(0) G(dx)
\end{equation}
Hence $S(0)$ serves as a ``density matrix'' for count rates. It is normalized to the total particle rate $\gamma$. For a discrete family of counters
with POVM elements $G_x\geq0$, $\sum_xG_x=\idty$ the arrival rate at counter $x$ is $\gamma_x=\tr S(0)G_x$.

The second moment has a density depending only on the time difference $\tau=t_2-t_1$. For discrete counters we also give the form of the normalized correlation function $g^{(2)}$:
\begin{eqnarray}\label{timemoment2}
   \factmom_1(dt_1\,dx_1,dt_2\,dx_2)&=&dt_1\,dt_2\,M_{t_1-t_2}(dx_1,dx_2)\\
   M_{\tau}(dx_1,dx_2)&=&(\tr S(0) G(dx_1))(\tr S(0) G(dx_2))\nonumber\\
                        &&\quad+s \tr S(\tau)^*G(dx_1)S(\tau)G(dx_2)\nonumber\\
   g^{(2)}_{xy}(\tau)&=& 1+\frac s{\gamma_x\gamma_y}\tr S(\tau)^*G_xS(\tau)G_y  \nonumber
\end{eqnarray}
Here we used the symmetry $S(-\tau)=S(\tau)^*$ to write the expression in a more obviously positive form.

It is not very enlightening to write down the higher moments. The third moment \eqref{factmomqf3} will contain contributions
$\tr S(t_3-t_1)G_xS(t_1-t_2)G_yS(t_2-t_3)G_z$.

\section{Examples}\label{sec:examples}
\subsection{Second order correlation function}
In the case of the free Hamiltonian in one dimension, we have $\mathcal{H}_E=\Cx^2$. Taking $V_E=\idty$ for $E\geq0$, the source \eqref{hatsiglimit}, and one detector $G$ corresponding to the detection of right going (positive momenta) particles, we get
\begin{align*}
\gamma &= (2\pi)^{-1}\chi(0), & g^{(2)}(\tau) &= 1+|\chi(\tau)|^2\chi(0)^{-2},
\end{align*}
where $\chi(\tau) = 4\pi\lambda\int_0^\infty dE\, e^{-i\tau E} |\phi_{E,+}|^2\,|1-\lambda h(E)|^{-2}$,
and $\phi_{E,+}$ is the positive momentum component of $\phi_E$.
% ---------------- FIG. BEGINS ----------------
\begin{figure}[t]
\begin{center}
\includegraphics[angle=0,width=0.9\linewidth]{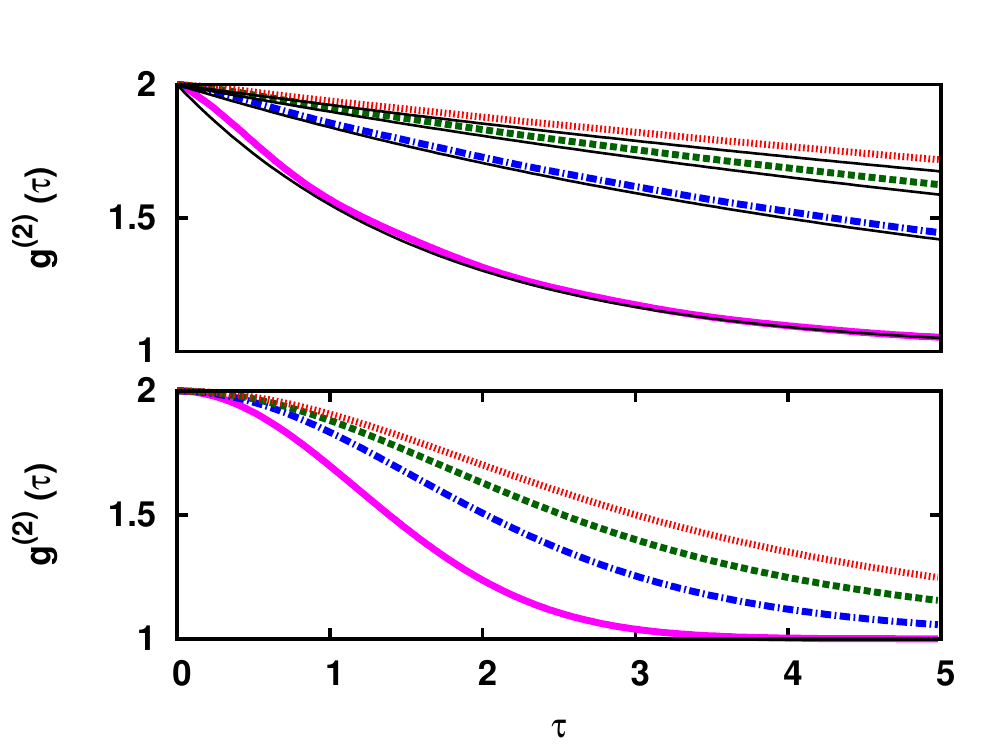}
\end{center}
\caption{\label{fig1} (Color online)
Second order normalized correlation function $g^{(2)} (\tau)$.
Upper figure: Lorentzian, lower figure: Gaussian initial spectral density.
Boltzmann model (thick solid line), quasi-free state given by
\eqref{hatsiglimit} with rates $\gamma=0.5,1.0,1.5$ (bottom to top curves).}
\end{figure}
% ---------------- END FIG. ----------------
In Fig.~\ref{fig1}, $g^{(2)} (\tau)$ is shown for a
Lorentzian $\fabsq{\phi_{E,+}} \sim (\alpha/2)/((E-E_0)^2 +
(\alpha/2)^2)/\pi$ and a Gaussian $\fabsq{\phi_{E,+}} \sim
\fexp{-(E-E_0)^2/(2\alpha^2)}/(\sqrt{2\pi}\alpha)$, with no negative momentum components.
In the case of the Lorentzian the correlation function can be approximated
for $E_0 \gg \alpha$ by $g^{(2)} \approx 1 + e^{-\tau(\alpha-2\lambda)}$ (see
the solid lines in Fig.~\ref{fig1}), i.e., as an exponential modified by the intensity parameter $\lambda$, as Bose statistical effects become more relevant.

\subsection{Plane wave beams}
The energy density of the beam is $\tr\hatsigma(E)$, which clearly needs to be
integrable. In the Fermi case the constraint $\hatsigma\leq\idty$ excludes
singularities in this density. However, in the Bose case, we can also consider
singular distributions.
Let us assume a one-dimensional, free Hamiltonian and we are restrict to only
positive momenta to simplify the notation.
In this case the beam state is given by
\begin{eqnarray*}
\hat\sigma = \int_0^\infty dE\, \alpha(E) \ketbra{E}{E}\,.
\end{eqnarray*}
where $\ket{E_}$ are the generalized energy eigenvectors.
Consider a sequence of functions $\alpha_n$ with $\alpha_n(E)
\stackrel{n\to\infty}\longrightarrow \kappa \delta (E-E_0) =: \alpha (E)$.
The characteristic function in this case is
\begin{eqnarray*}
C(f) = \left(1-\kappa \bra{E_0}F[e^{if}-1]\ket{E_0}\right)^{-1}\,.
\end{eqnarray*}
The rate is then given by $\gamma_Q = \kappa \bra{E_0}F[f]\ket{E_0}$.
We can get the number distributions $p_n$ for
a measurement result in an interval $Y$ from the characteristic function,
see \eqref{numberdist}.
This characteristic function for for $\alpha (E) = \kappa \delta (E-E_0)$ is
\begin{eqnarray}
C(\lambda \chi_Y)&=&\frac{1}{1-(e^{i\lambda}-1)\kappa\bra{E_0}F(Y)\ket{E_0}}\nonumber\\
&=& \sum_{n=0}^\infty \frac{q^n}{(1+q)^{n+1}} e^{i\lambda n}
\label{num2}
\end{eqnarray}
where $q= \kappa\bra{E_0}F(I)\ket{E_0}$. By comparing \eqref{num2} and
\eqref{numberdist}, we get the number
distribution for a detection in the interval $Y$, namely
$p_{Q,n} = \frac{q^n}{(1+q)^{n+1}}$.
Fig.~\ref{fig2} shows examples of the number distribution for an arrival-time
measurement with $q=\sqrt{10}$.

It is illustrative to compare it with the number statistics
of a coherent beam given by \eqref{cfcoherent}
with $\phi (E) = \sqrt{\kappa} \delta (E-E_0)$.
The characteristic function
is then $C(f) = \exp\left(\kappa \bra{E_0}|F[e^{if}-1] \ket{E_0}\right)$
and from this we get the rate $\gamma_C = \kappa
\bra{E_0}F[f]\ket{E_0}=\gamma_Q$ and a Poisson number distribution
$p_{C,n}= \frac{1}{n!} q^n e^{-q}$
which is also shown in Fig.~\ref{fig2}.
% ---------------- FIG. BEGINS ----------------
\begin{figure}[t]
\begin{center}
\includegraphics[angle=0,width=0.9\linewidth]{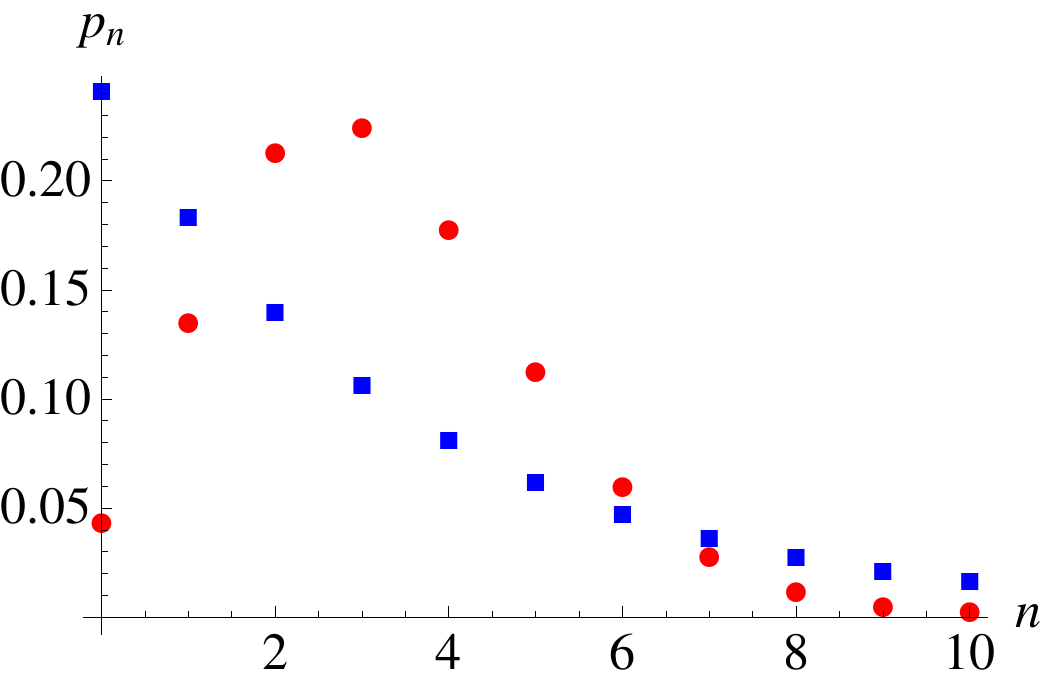}
\end{center}
\caption{\label{fig2} (Color online)
Particle number distribution $p_n$ with $q=\sqrt{10}$,
quasi-free beam $p_{Q,n}$ (blue boxes),
coherent beam $p_{C,n}$ (red circles).}
\end{figure}
% ---------------- END FIG. ----------------

Let us look at a Kijowski's arrival time
measurement in more detail. In that case, we get the rate $\gamma_K =
\frac{\kappa}{2\pi}$ and $q_K = \gamma_K \,l(Y)$ where $l(Y)=\int_Y dt$.
For a quasi-free beam we get for the probability for no detection in
an interval $[t_1,t_2]$ is $p_{Q,0}(t_1, t_2) = 1/(1+\gamma_K (t_2-t_1))$
and for a coherent beam $p_{C,0}= e^{-\gamma_K (t_2, t_1)q}$.
The waiting time calculated by \eqref{nextclickdensity}
is now for a quasi-free beam
\begin{equation}
w_{Q} (\tau)=\frac{2\gamma_K}{(1+\gamma_K \tau)^3}
\end{equation}
and for a coherent beam
\begin{equation}
w_C (\tau)= \gamma_K e^{-\gamma_K\tau}
\end{equation}
which is -as expected- an exponential distribution.

\section*{Acknowledgment} The authors acknowledge support by the BMBF (Ephquam
project), the EU project CORNER, and the Academy of Finland.

\end{document}